\documentclass[12pt]{article}

\usepackage[T1]{fontenc}
\usepackage[utf8]{inputenc}
\usepackage{amsmath,amssymb,amsthm,mathtools}
\usepackage[hidelinks]{hyperref}

\newtheorem{proposition}{Proposition}

\newtheorem{lemma}{Lemma}
\newtheorem{theorem}{Theorem}

\newcommand{\Hplus}{\mathbb{C}_+}
\newcommand{\eps}{\varepsilon}
\DeclareMathOperator{\adj}{adj}

\newcommand{\R}{\mathbb{R}}
\newcommand{\C}{\mathbb{C}}
\newcommand{\Sn}{\mathbb{S}^{n}}
\newcommand{\TA}{\mathcal{T}_{A}}
\newcommand{\TAs}{\mathcal{T}_{A}^{*}}
\newcommand{\TAh}{\hat{\mathcal{T}}_{A}}
\newcommand{\TAhs}{\hat{\mathcal{T}}_{A}^{*}}

\newcommand{\TAt}{\mathcal{T}_{\tilde{A}}}

\newcommand{\beq}{\begin{equation}}
\newcommand{\eeq}{\end{equation}}

\newcommand{\entrygeq}{\geq}

\title{A Duality Reformulation of the Companion-Matrix Lyapunov Problem}
\author{Augusto Ferrante
\thanks{A. Ferrante is with the Dipartimento di Ingegneria dell'Informazione,
 Universit\`a di Padova,
    via Gradenigo 6B, 35131 Padova, Italy, e-mail: {\texttt augusto@dei.unipd.it}}
   }

\date{}

\begin{document}

\maketitle

\begin{abstract}
We study the relation between positive semidefiniteness and entrywise
nonnegativity for solutions of continuous-time Lyapunov equations. For a
real Hurwitz matrix $A$, we show that the solution of
$AP+PA^\top=-Q$ is positive semidefinite for every symmetric entrywise
nonnegative $Q$ if and only if the solution of $A^\top X+XA=-R$ is
entrywise nonnegative for every positive semidefinite $R$. This
equivalence follows from the adjointness of the two solution operators
and extends to all real unmixed matrices. We then prove that both
properties hold for every real Hurwitz companion matrix, settling a
conjecture previously established under the additional assumption of a
real spectrum. The proof combines a controllability-Gramian normalization
with a pairwise positivity theorem for the coefficients of the adjugate
polynomial of a real accretive matrix. The latter is obtained from a
coefficient-sign property of bivariate polynomials, an auxiliary
determinant that does not vanish on the product of two open right
half-planes, and a rank-one perturbation argument. Covariance and energy
interpretations connect these results with dissipative realizations,
damped second-order systems, and comparisons between input Gramians.
\end{abstract}

\section{Introduction}

Lyapunov equations, together with their discrete-time counterpart
known as Stein equations, play a central role in almost all the areas of systems and control theory including optimal control \cite{AndersonMoore2007}, optimal filtering \cite{AndersonMoore1979}, spectral factorization and stochastic realization \cite{LPBook,Ferrante-94-ieee}, passivity theory \cite{Anderson-Vongpanitlerd}, and stabilization \cite{BoydEtAl1994} to mention but a few. 
Solutions of Lyapunov equations describe quadratic
energy functions, controllability and observability Gramians, and
stationary covariances of linear stochastic systems. A basic result is
that, when $A\in\R^{n\times n}$ is Hurwitz, the unique solution of
\begin{equation}\label{eq:intro-adjoint}
  A^\top X+XA=-R
\end{equation}
is positive semidefinite whenever $R$ is positive semidefinite. This
property follows directly from the integral representation
\[
  X=\int_0^\infty e^{A^\top t}R e^{At}\,dt.
\]
Entrywise nonnegativity of $X$ is a separate question. It imposes a sign
condition on each matrix entry and is not implied by positive
semidefiniteness. Its validity depends on the structure of $A$ and on the
coordinates in which the equation is expressed.

The companion-matrix Lyapunov problem asks whether the solution of
\eqref{eq:intro-adjoint} is entrywise nonnegative for every
$R\succeq0$ when $A$ is a real Hurwitz companion matrix. We use the
companion convention with ones on the superdiagonal and the negatives
of the characteristic-polynomial coefficients in the last row, as in
\eqref{comp-matrix}. In \cite{Ferrante2026}, this property was proved
under the additional assumption that all eigenvalues of $A$ are real.
The case allowing nonreal complex-conjugate eigenvalue pairs was left
open. The main result of the present paper establishes the property
for every real Hurwitz companion matrix. No assumption of a real or
simple spectrum, or of diagonalizability, is required.

Our approach begins with a duality relation between
\eqref{eq:intro-adjoint} and the Lyapunov equation
\begin{equation}\label{eq:intro-primal}
  AP+PA^\top=-Q.
\end{equation}
For a fixed Hurwitz matrix $A$, the corresponding solution operators
$\mathcal T_A$ and $\mathcal T_A^*$ are adjoint with respect to the
Frobenius inner product on real symmetric matrices. The cones of
positive semidefinite matrices and symmetric entrywise nonnegative
matrices are each self-dual for this inner product. This 
provides a general relation between the two notions of positivity.
The duality remains valid whenever the Lyapunov operator is invertible,
or equivalently when $A$ is unmixed namely, no two eigenvalues of $A$, including
a repeated choice of the same eigenvalue, sum to zero.

However, the question remains of determining the class of state matrices $A$, if any, for which 
(\ref{eq:intro-adjoint}) maps positive semidefinite matrices $R$ into entrywise nonegative solutions
$X$ or, equivalently, 
for which 
(\ref{eq:intro-primal}) maps entrywise nonnegative matrices  $Q$ into  positive semidefinite  solutions
$P$. This appears to be a very hard question. In this paper we first provide 
a subclass of these matrices: the set of Hurwitz matrices in companion form.
Building on this result we show, on the one hand,  that if $A$ is not Hurwitz the properties cannot hold and, on the other hand,  that the two properties hold
for a class of state matrices which is much larger than that of Hurwitz matrices in companion form: the matrices obtained from a Hurwitz matrix in companion form by any similarity transformation induced by any entrywise nonnegative nonsingular matrix.  
We also exhibit a matrix outside of this larger class for which the properties hold. 
Further research thus is needed to characterize precisely the class of all matrices for which the above properties hold.

The paper is organized as follows. Section \ref{sec:aux-res} develops the polynomial and adjugate-coefficient auxiliary results which are interesting {\em per se}. Section \ref{sec:duality-lyap} establishes the duality
between the two positivity properties and its extension to unmixed
matrices. Section \ref{sec:main-res} proves the companion-matrix theorem. 
Section \ref{sec:main-res-ext} discusses the extension of the class of state matrices $A$ for which (\ref{eq:intro-adjoint}) maps positive semidefinite matrices $R$ into entrywise nonnegative solutions. Section \ref{sec:interpreta}
discusses the control and filtering interpretations.
Finally, some conclusions are drawn in Section \ref{sec:conclusions}.

{\bf Notation.}
Throughout this paper, \(\Sn\) denotes the space of real symmetric
\(n\times n\) matrices, equipped with the Frobenius inner product
\[
  \langle X,Y\rangle:=\operatorname{tr}(XY).
\]
For $X\in\Sn$, the notation \(X\entrygeq0\) (\(X> 0\)) means that \(X\) is entrywise nonnegative
(entrywise positive),
whereas \(X\succeq 0\) (\(X\succ 0\)) denotes positive semidefiniteness (positive definiteness).
For a complex matrix $P$, $P^\top$ and $P^*$ denote respectively the transpose and  the conjugate transpose of $P$. If $P$ is a square matrix, $\adj(P)$ denotes the adjugate of $P$. For a complex number $z$, $\Re[z]$ and $\Im[z]$ denote respectively the real part and the imaginary part of $z$.
We also define $\Hplus:=\{z\in\C:\ \Re[z]>0\}$.
A real matrix $M$ is called {\em accretive} if $M+M^\top\succeq0$.

\section{Auxiliary results}\label{sec:aux-res}

First we establish a Lemma which seems to be more or less common knowledge in the literature.
However, we have only seen a similar statement in a paper that has only appeared in Arxiv \cite[Fact 4]{Fisk}: the proof reported there is so terse that we have not been able to confirm it.
\begin{lemma}\label{lem:coefficient-sign}
Let
\[
P(s,t)=\sum_{i,j=0}^\ell c_{ij}s^it^j\in\R[s,t]
\]
be a polynomial. Suppose
\begin{equation}\label{eq:nonvanishing}
P(s,t)\ne0\qquad\text{for all }s,t\in\Hplus.
\end{equation}
Then all nonzero coefficients of $P$ have the same sign.
This sign is the same as that of $P(1,1)=\sum_{i,j=0}^\ell c_{ij}$.
\end{lemma}

\begin{proof}
We first establish two elementary one-variable facts.
These facts are indeed well-known but we provide a proof for completeness.

\noindent\textit{Fact 1: Property of Lemma \ref{lem:coefficient-sign} holds for real one-variable polynomials.}
Let $f=\sum_{i=0}^n f_iz^i\in\R[z]$ and suppose $f(z)\ne0$ for $z\in\Hplus$.
Then, every root of $f$ lies in the closed left half-plane and factoring over
$\R$, we may write
\begin{equation}\label{eq:real-factorization}
f(z)=\kappa
\prod_{k=1}^{r}(z+\alpha_k)
\prod_{l=1}^{m}
\bigl(z^2+2\beta_l z+\omega_l^2\bigr),
\end{equation}
where $\kappa\in\R\setminus\{0\}$,
$\alpha_k\ge0$, $\beta_l\ge0$, and $\omega_l>0$.
Hence, $f_i=\kappa s_i$, where $s_i$ is nonnegative because it is given by a sum of products of the nonnegative terms $\alpha_k\ge0$, $\beta_l\ge0$, and $\omega_l>0$.

\medskip

\noindent\textit{Fact 2: Differentiation preserves nonvanishing region (Gauss--Lucas Theorem).}
Let $f=\sum_{i=0}^n f_iz^i\in\C[z]$ have degree $n\ge1$ and no zeros in $\Hplus$.
Write its roots, repeated with multiplicity, as $\rho_1,\ldots,\rho_n$.
For $z\in\Hplus$,
\[
\frac{f'(z)}{f(z)}=\sum_{i=1}^{n}\frac{1}{z-\rho_i}.
\]
Since $\Re\rho_i\le0$,
\[
\Re\!\left(\frac{1}{z-\rho_i}\right)
=\frac{\Re z-\Re\rho_i}{|z-\rho_i|^2}>0.
\]
The sum therefore has strictly positive real part. In particular: $f'(z)\ne0$ on $\Hplus$. Repeating this argument shows that
\begin{equation}\label{eq:derivatives-stable}
f^{(p)}(z)\ne0\qquad(z\in\Hplus,\ 0\le p\le n).
\end{equation}

\medskip

Before addressing the bivariate case we normalize the sign.
Since $P$ is real,
continuous, and nonzero on the connected set $(0,\infty)^2$,  its
sign is constant there. 
If this sign is negative we can replace $P$ by
$-P$, which has the opposite sign and opposite coefficients.
Therefore, it is enough to prove that all the coefficients are nonnegative under the assumption
\begin{equation}\label{eq:positive-orthant}
P(u,v)>0\qquad\text{for all real }u,v>0.
\end{equation}

\medskip

If $P$ is constant, the conclusion is immediate. Hence we can safely assume that
$d\ge1$ is the total degree of $P$. Let the  homogeneous part of total degree $d$ be
\[
H_d(s,t)=\sum_{i+j=d}c_{ij}s^it^j\not\equiv 0.
\]
The one-variable polynomial
\[
\gamma(\varepsilon):=H_d(1,\varepsilon)
=\sum_{j=0}^{d}c_{d-j,j}\varepsilon^j
\]
is not identically zero. It has only finitely many real zeros.
Consequently, there is $\varepsilon_0>0$ such that
$\gamma(\varepsilon)\ne0$ whenever $0<\varepsilon<\varepsilon_0$.

Fix such an $\varepsilon$ and define
\begin{equation}\label{eq:linear-change}
P_\varepsilon(s,t):=P(s,t+\varepsilon s).
\end{equation}
If $s,t\in\Hplus$, then $t+\varepsilon s\in\Hplus$. Thus
\begin{equation}\label{eq:changed-stability}
P_\varepsilon(s,t)\ne0\qquad(s,t\in\Hplus).
\end{equation}
Also, $P_\varepsilon(u,v)>0$ for real $u,v>0$, by
\eqref{eq:positive-orthant}.

Expanding \eqref{eq:linear-change} shows that
\begin{equation}\label{eq:constant-leading}
P_\varepsilon(s,t)
=\gamma(\varepsilon)s^d+\sum_{p=0}^{d-1}b_{\varepsilon,p}(t)s^p,
\qquad b_{\varepsilon,p}\in\R[t].
\end{equation}
Indeed, a term $c_{ij}s^i(t+\varepsilon s)^j$ can contribute to $s^d$
only if $i+j=d$, in which case its contribution is
$c_{ij}\varepsilon^j$. In particular, the coefficient of $s^d$ is
independent of $t$ and nonzero. Therefore
$P_\varepsilon(\,\cdot\,,t)$ has degree exactly $d$ for every
$t\in\C$; no specialization can lower that degree.

\medskip

Now we prove that the  coefficients of $P_\varepsilon(s,t)$
are nonnegative for all $0<\varepsilon<\varepsilon_0$.
Fix $p\in\{0,\ldots,d\}$ and a real number $\delta>0$. Define
\begin{equation}\label{eq:q-definition}
q_{\varepsilon,\delta,p}(t)
:=\left.\frac{\partial^p}{\partial s^p}P_\varepsilon(s,t)
\right|_{s=\delta}\in\R[t].
\end{equation}
For any fixed $t\in\Hplus$, the polynomial
$s\mapsto P_\varepsilon(s,t)$ has degree $d$ and no zeros in
$\Hplus$, by \eqref{eq:constant-leading} and
\eqref{eq:changed-stability}. Fact~2, evaluated at $s=\delta$,
therefore gives
\begin{equation}\label{eq:q-nonvanishing}
q_{\varepsilon,\delta,p}(t)\ne0
\qquad\text{for every }t\in\Hplus.
\end{equation}

Next take a real $v>0$. The real polynomial
$s\mapsto P_\varepsilon(s,v)$ has no zeros in $\Hplus$ and is positive
for real $s>0$. By Fact~1 all its coefficients are nonnegative.
Its leading coefficient is nonzero and hence strictly positive.
It follows that
\begin{equation}\label{eq:q-positive}
q_{\varepsilon,\delta,p}(v)>0
\qquad\text{for every real }v>0.
\end{equation}

Apply Fact~1 again, now to the real polynomial
$q_{\varepsilon,\delta,p}$ in the variable $t$.
Equations \eqref{eq:q-nonvanishing} and \eqref{eq:q-positive} show
that every coefficient of this polynomial is nonnegative.

Write
\[
P_\varepsilon(s,t)=\sum_{i=0}^{d}\sum_{j=0}^{\ell}
c_{ij}^{(\varepsilon)}s^it^j.
\]
For each $j$, the coefficient of $t^j$ in
\eqref{eq:q-definition} is
\[
\sum_{i=p}^{d}\frac{i!}{(i-p)!}
c_{ij}^{(\varepsilon)}\delta^{i-p}\ge0.
\]
Letting $\delta\to 0^+$ in this finite sum yields
$p!\,c_{pj}^{(\varepsilon)}\ge0$ so that
\[
c_{pj}^{(\varepsilon)}\geq 0,\qquad \forall \varepsilon\in(0,\varepsilon_0).
\]
Since $p$ and $j$ were arbitrary, every coefficient of $P_\varepsilon$
is nonnegative. 

To reach the conclusion, it is now sufficient to 
compute the binomial expansion of $P(s,t+\varepsilon s)$ 
and observe that
$c_{ij}^{(\varepsilon)}$ is a polynomial in $\varepsilon$ of the form
\begin{equation}
c_{ij}^{(\varepsilon)}=c_{i,j}+\varepsilon c_{i-1,j+1}\binom{j+1}{j}+\dots + \varepsilon^l c_{i-l,j+l}\binom{j+l}{j},
\end{equation}
where $l:=\min\{i,\ell-j\}$,  so that
\[
\lim_{\varepsilon\to 0^+}c_{ij}^{(\varepsilon)}=c_{ij}.
\]
Thus every coefficient of the  polynomial $P$
is nonnegative.  
If we eliminate the sign normalization, we conclude that all the nonzero coefficients $c_{ij}$ of $P$ have the same sign.
The fact that this sign coincides with the sign of 
$P(1,1)=\sum_{i,j=0}^\ell c_{ij}$ is obvious.
\end{proof}

The following result, which uses Lemma \ref{lem:coefficient-sign}, is one of the cornerstones for the proof of one of our main results.

\begin{lemma}\label{lem:pairwise-adjugate}
Let $M\in\R^{n\times n}$ satisfy $M+M^\top\succeq0$, and write
\begin{equation}\label{eq:F-polynomial}
F(z):=\adj(zI+M)=\sum_{p=0}^{n-1}z^pF_p.
\end{equation}
Then
\begin{equation}\label{eq:pairwise-adjugate}
F_pF_q^\top+F_qF_p^\top\succeq0
\qquad(0\le p,q\le n-1).
\end{equation}
\end{lemma}

\begin{proof}
Fix $x\in\R^n$ and $\tau\ge0$, and introduce
\begin{equation}\label{eq:D-tau}
D_\tau(s,t):=
\det\!\left((sI+M^\top)(tI+M)+\tau xx^\top\right).
\end{equation}

\medskip

We first show that
\begin{equation}\label{eq:bivariate-stability}
D_\tau(s,t)\ne0
\qquad\text{whenever }\Re s>0\text{ and }\Re t>0.
\end{equation}
Suppose, by way of contradiction, that this is not the case. Then, there  exist
$s,t\in\Hplus$ and $z\in\C^n$ with $z\neq 0$ such that
\[
(sI+M^\top)(tI+M)z+\tau xx^\top z=0.
\]
By multiplying this equation on the left by $z^*$, we get
\begin{equation}\label{eq:energy-identity}
\begin{aligned}
z^*(sI-\bar{t}I+\bar{t}I+M^\top)(tI+M)z+\tau z^*xx^\top z&=&
(s-\bar{t})z^*(tI+M)z\\
&&+\|(tI+M)z\|^2+\tau |x^\top z|^2\\
&=0.
\end{aligned}
\end{equation}
Notice that, by accretivity of $M$ we have
\begin{equation}\label{eq:prod-positive}
\Re[z^*(tI+M)z]=\Re[t]\|z\|^2 +\frac{1}{2}z^*(M+M^\top)z\geq \Re[t]\|z\|^2>0, 
\end{equation}
where the last inequality, follows from $t\in\Hplus$ and  $z\neq 0$.
In particular, $z^*(tI+M)z\ne0$. Consequently, we can rewrite 
\eqref{eq:energy-identity} as
\begin{equation}\label{eq:energy-identity-2}
s=\bar{t}-\frac{\|(tI+M)z\|^2+\tau |x^\top z|^2}{z^*(tI+M)z}.
\end{equation}
The latter, together with \eqref{eq:prod-positive} gives
\begin{equation}
\begin{aligned}
\Re[s]&
 =\Re[t]-\Re\left[\overline{z^*(tI+M)z}\right]\frac{\|(tI+M)z\|^2+\tau |x^\top z|^2}{|z^*(tI+M)z|^2}\\
 &
 =\Re[t]-\Re[z^*(tI+M)z]\frac{\|(tI+M)z\|^2+\tau |x^\top z|^2}{|z^*(tI+M)z|^2}
 \\
 &
 \leq
 \Re[t]-\Re[t]\frac{\|z\|^2\left(\|(tI+M)z\|^2+\tau |x^\top z|^2\right)}{|z^*(tI+M)z|^2}\\
 &
 \leq
 \Re[t]-\Re[t]\frac{\|z\|^2\|(tI+M)z\|^2}{|z^*(tI+M)z|^2}\leq 0,
 \end{aligned}
\end{equation}
where the last inequality is a direct consequence of the
Cauchy--Schwarz inequality.
This is a contradiction because we are assuming $\Re[s]>0$.
Hence,
\eqref{eq:bivariate-stability} holds.

\medskip

Now, we are in position to use Lemma \ref{lem:coefficient-sign}. In fact, 
for each $\tau\geq 0$
$D_\tau(s,t)$ is a polynomial with real coefficients and 
we have just seen that it does not vanish for $s,t\in\Hplus$.
Hence, all the nonzero coefficients of $D_\tau(s,t)$ have the same sign.
We now argue that this sign is positive.
Indeed, for $s=t=1$, 
$
D_\tau(1,1)=
\det\!\left((I+M^\top)(I+M)+\tau xx^\top\right).
$
Since $M$ is accretive and $\tau\geq 0$, we have
$(I+M^\top)(I+M)+\tau xx^\top\succeq I$, so that
$D_\tau(1,1)\geq \det(I)=1.$

\medskip

The final step uses the so-called rank-one determinant identity and
the fact that
$\adj(BC)=\adj(C)\adj(B)$.
Let $a(z):=\det(zI+M)=\sum_{j=0}^n a_jz^j$. 
We have 
\begin{align}
D_\tau(s,t)
&=a(s)a(t)
+\tau x^\top\adj\!\left((sI+M^\top)(tI+M)\right)x\notag\\
&=a(s)a(t)+\tau x^\top F(t)F(s)^\top x.
\label{eq:rank-one-expansion}
\end{align}
For $0\le p,q\le n-1$,
the coefficient of $s^pt^q$ in \eqref{eq:rank-one-expansion} is
\[
a_pa_q+\tau x^\top F_qF_p^\top x\ge0,
\qquad\text{for every }\tau\ge0.
\]
This expression is affine in $\tau$. Therefore its slope is nonnegative:
\[
x^\top F_qF_p^\top x\ge0.
\]
Since
\[
x^\top F_qF_p^\top x
=\frac12x^\top
\left(F_pF_q^\top+F_qF_p^\top\right)x
\]
and $x\in\R^n$ was arbitrary, the lemma follows.
\end{proof}

\section{Lyapunov equations inducing a connection between two types of positivity}
\label{sec:duality-lyap}

Next we consider two Lyapunov equations
connected by duality. We show that these equations could provide an interesting relation between two notions of positivity.
This result, besides being instrumental for proving 
a conjecture that was only partially addressed in \cite{Ferrante2026},
is relevant {\em per se} as discussed in Section \ref{sec:interpreta}.

Let \(A\in\R^{n\times n}\) be Hurwitz and $Q,R\in\Sn$.
Consider the two Lyapunov equations
\begin{equation}\label{eq:lyap-primal}
  AP+PA^{\top}=-Q.
\end{equation}
and
\begin{equation}\label{eq:lyap-adjoint}
  A^{\top}X+XA=-R.
\end{equation}
We have the following result.
\begin{theorem}\label{theo:duality}
Let \(A\in\R^{n\times n}\) be Hurwitz and $Q,R\in\Sn$.
Then the solution $P$ of (\ref{eq:lyap-primal}) is positive semidefinite
for all $Q\geq 0$ if and only if the solution $X$ of (\ref{eq:lyap-adjoint})
is entrywise nonnegative for all $R\succeq0$.
\end{theorem}
\begin{proof}
Define the linear operator
\begin{equation}\label{eq:TA}
  \TA(Q)
  :=
  \int_{0}^{\infty}e^{At}Qe^{A^{\top}t}\,dt,
  \qquad Q\in\Sn.
\end{equation}
Then \(P=\TA(Q)\) is the unique solution of
(\ref{eq:lyap-primal}).
Define also the dual linear operator
\begin{equation}\label{eq:TA-adjoint}
  \TAs(R)
  :=
  \int_{0}^{\infty}e^{A^{\top}t}Re^{At}\,dt,
  \qquad R\in\Sn,
\end{equation}
so that \(X=\TAs(R)\) is the unique solution of
(\ref{eq:lyap-adjoint}).

Notice that the operator $\TAs$ is the adjoint of \(\TA\) with respect to the Frobenius inner product.
In fact, for every \(Q,R\in\Sn\), we have
\begin{equation}\label{eq:adjoint-identity}
\begin{aligned}
  \langle R,\TA(Q)\rangle&=
  \operatorname{tr}\left(R\int_{0}^{\infty}e^{At}Qe^{A^{\top}t}\,dt\right)\\
&=
 \int_{0}^{\infty}\operatorname{tr}\left(Re^{At}Qe^{A^{\top}t}\right)\,dt\\
&=
 \int_{0}^{\infty}\operatorname{tr}\left(e^{A^{\top}t} Re^{At}Q\right)\,dt\\
&=\operatorname{tr}\left(\int_{0}^{\infty}e^{A^{\top}t}Re^{At}\,dt\  Q\right)\\
&=
  \langle\TAs(R),Q\rangle.
  \end{aligned}
\end{equation}
Thus, we have to prove that the following properties are
equivalent:
\begin{align}
R=R^\top\succeq0
  &\quad\Longrightarrow\quad
  \TAs(R)\entrygeq0,
  \label{eq:property-adjoint}\\
  Q=Q^\top\entrygeq0
  &\quad\Longrightarrow\quad
  \TA(Q)\succeq0.
  \label{eq:property-primal}
\end{align}

Assume first that \eqref{eq:property-adjoint} holds. Then
\begin{equation}
\TAs(xx^{\top})\entrygeq0,\qquad \forall x\in\R^n.
\end{equation}
As a consequence, for all $x\in\R^n$ and $Q=Q^\top\entrygeq0$, we have:
\[
  x^{\top}\TA(Q)x
  =
  \operatorname{tr}(xx^{\top}\TA(Q))=\langle xx^{\top},\TA(Q)\rangle
  =\langle \TAs(xx^{\top}),Q\rangle\geq 0\]
which proves \(\TA(Q)\succeq0\) i.e. \eqref{eq:property-primal}.

Conversely, assume \eqref{eq:property-primal}. Fix \(x\in\R^n\) and denote by \(E_{ij}\in\R^{n\times n}\) the matrix whose
only nonzero entry is a \(1\) in position \((i,j)\). Choosing
\(Q=E_{ii}\entrygeq0\), \eqref{eq:property-primal} gives
\[
  [\TAs(xx^{\top})]_{ii}
  =
  \langle \TAs(xx^{\top}),E_{ii}\rangle
  =
  x^{\top}\TA(E_{ii})x
  \geq0.
\]
For \(i\neq j\), choose the symmetric entrywise nonnegative matrix
\(Q=E_{ij}+E_{ji}\). Then, \eqref{eq:property-primal} gives
\[
  [\TAs(xx^{\top})]_{ij}
  =\frac{1}{2}
  \langle \TAs(xx^{\top}),E_{ij}+E_{ji}\rangle
  =\frac{1}{2}
  x^{\top}\TA(E_{ij}+E_{ji})x
  \geq0.
\]
Therefore \(\TAs(xx^{\top})\entrygeq0\) for every \(x\).

Finally, by suitably selecting $r\leq n$ vectors $x_m \in\R^n$, every \(R\succeq0\) can be written as
\begin{equation}\label{R-sum}
  R=\sum_{m=1}^{r}x_mx_m^{\top}.
\end{equation}
By linearity,
\[
  \TAs(R)=\sum_{m=1}^{r}\TAs(x_mx_m^\top)\entrygeq0
\]
which proves \eqref{eq:property-adjoint}.
\end{proof}

\subsection{Removing Hurwitz assumption}
Notice that the assumption of $A$ being Hurwitz in Theorem \ref{theo:duality} can be considerably relaxed. Indeed, the result of this proposition holds as long as 
the Lyapunov equations (\ref{eq:lyap-primal}) and (\ref{eq:lyap-adjoint}) admit a (necessarily unique) solution for any $Q$ and $R$ or, equivalently, if $\lambda_1+\lambda_2\neq 0$ for any pair of
(not necessarily distinct) eigenvalues $\lambda_1,\lambda_2\in\sigma(A)$. 
In this case the matrix $A$ is called {\emph{unmixed}}.
To see this, it is sufficient to observe that if $A$ is unmixed
the solutions $P$ and $X$ of \eqref{eq:lyap-primal}
and \eqref{eq:lyap-adjoint} can respectively be expressed by the linear operators
\begin{equation}\label{eq:TAh}
  \TAh(Q)
  :=
  {\rm vec}^{-1}[(I\otimes A + A \otimes I)^{-1}{\rm vec}(-Q)],
  \qquad Q\in\Sn,
\end{equation}
and
\begin{equation}\label{eq:TAh-adjoint}
  \TAhs(R)
  :=
  {\rm vec}^{-1}[(I\otimes A^\top + A^\top \otimes I)^{-1}{\rm vec}(-R)],
  \qquad R\in\Sn,
\end{equation}
where $\otimes$ denotes the Kronecker product and ${\rm vec}$ is the linear operator mapping a matrix in $W\in\Sn$ into a
vector $w\in\R^{n^2}$ obtained by stacking the columns of $W$ one on top of the other. 
Thus, we only need to show that 
$\TAhs$ is the adjoint of \(\TAh\) with respect to the Frobenius inner product.
In fact, if this is the case, we can repeat verbatim the proof of Theorem \ref{theo:duality}.
The fact that $\TAhs$ is the adjoint of \(\TAh\) is a straightforward calculation.
Indeed, for any $Q,R\in\Sn$ we have
\begin{equation}\label{duality-deneral}
\begin{aligned}
\langle \TAh(Q),R\rangle&=\operatorname{tr}\left(\TAh(Q)R\right)
\\
&=[{\rm vec}(\TAh(Q))]^\top{\rm vec}(R)
\\
&=[{\rm vec}(-Q)]^\top(I\otimes A^\top + A^\top \otimes I)^{-1}{\rm vec}(R)
\\
&=-[{\rm vec}(Q)]^\top(I\otimes A^\top + A^\top \otimes I)^{-1}{\rm vec}(R)
\\
&=[{\rm vec}(Q)]^\top(I\otimes A^\top + A^\top \otimes I)^{-1}{\rm vec}(-R)
\\
&=[{\rm vec}(Q)]^\top{\rm vec}(\TAhs(R))\\
&=\langle Q,\TAhs(R)\rangle.
\end{aligned}
\end{equation}
In conclusion, the following stronger version of Theorem \ref{theo:duality} holds.
\begin{theorem}\label{prop:duality-stronger}
Let \(A\in\R^{n\times n}\) be unmixed and $Q,R\in\Sn$.
Then the solution $P$ of (\ref{eq:lyap-primal}) is positive semidefinite
for all $Q\geq 0$ if and only if the solution $X$ of (\ref{eq:lyap-adjoint})
is entrywise nonnegative for all $R\succeq0$.
\end{theorem}

We now show that the class of non-Hurwitz matrices for which 
 \eqref{eq:property-primal} holds
is empty. 
Let $A$ be non-Hurwitz and unmixed.   Consider a matrix $Q=Q^\top$ which is both (strictly) positive definite and entrywise nonnegative; for example $Q=I$. Let 
$P$ be the corresponding solution of \eqref{eq:lyap-primal} and assume, by way of contradiction,
that \eqref{eq:property-primal} holds.
Then $P\succeq 0$; moreover,  
\beq\label{contrad}
AP+PA^\top=-Q\preceq 0.
\eeq
Let $w$ be an eigenvector of $A^\top$ associated with an eigenvalue $\lambda$ with nonnegative real part ($\lambda$ exists because $A$ is non-Hurwitz). By multiplying \eqref{contrad} by $w$ on the right side and by 
$w^*$ on the left side, we get:
$w^*APw+w^*PA^\top w=
(\lambda +\overline{\lambda})w^*Pw
=2\Re[\lambda]w^*Pw
=-w^* Q w < 0$ which is a contradiction because  $2\Re[\lambda]\geq0$ and
$P\succeq 0$.

Theorem \ref{prop:duality-stronger} allows to conclude that 
also the class of non-Hurwitz matrices for which \eqref{eq:property-adjoint} holds
is empty which is not apparent by a direct analysis.

\section{A class of matrices $A$ for Theorem \ref{theo:duality}}
\label{sec:main-res}

The property in \eqref{eq:property-adjoint} is 
studied in the recent work \cite{Ferrante2026} for the case of 
\begin{equation}\label{comp-matrix}
A:=
\begin{bmatrix}
0&1&0&\cdots&0\\
0&0&1&\ddots&\vdots\\
\vdots&&\ddots&\ddots&0\\
0&\cdots&0&0&1\\
-\sigma_n&-\sigma_{n-1}&\cdots&-\sigma_2&-\sigma_1
\end{bmatrix}\in \R^{n\times n}
\end{equation}
being  companion associated with the real  Hurwitz polynomial
\begin{equation}\label{car-poly}
a(s)=s^n+\sigma_1s^{n-1}+\cdots+\sigma_n.
\end{equation}
That work proves that \eqref{eq:property-adjoint} holds
when the companion matrix \(A\) has only real eigenvalues, while the general case in which $A$ may have
complex-conjugate eigenvalue pairs remains an open conjecture.
Next we prove that such a conjecture is indeed a theorem by resorting to Theorem~\ref{theo:duality}.
To this aim, we first prove that 
Property (\ref{eq:property-primal})  holds for
$A$ given by \eqref{comp-matrix} and  $Q$ having the simple form  
\beq\label{elementary-Q}
Q=
Q_{kj}:=\varepsilon_k \varepsilon_l^\top+\varepsilon_l \varepsilon_k^\top,
\eeq
 where $\varepsilon_k$ denotes $k$-th canonical vector and $k,l=1,\dots,n$.

\begin{proposition}\label{prop:lyapunov-primal-sc}
Let $A$  be a real Hurwitz  companion matrix given by 
\eqref{comp-matrix}.
For $k,l\in\{1,\ldots,n\}$, let $Q_{kl}$ be given by \eqref{elementary-Q}
and 
$$
P_{kl}:=\int_0^\infty e^{At}Q_{kl}e^{A^\top t}\,dt
$$
be the corresponding solution of \eqref{eq:lyap-primal}.
Then
\begin{equation}\label{eq:P-positive}
 P_{kl}\succeq 0\qquad\text{for every } k,l=1,\dots, n.
\end{equation}
\end{proposition}

\begin{proof}
Define the Horner polynomials associated with the characteristic polynomial $a(s)$ of $A$:
$$
\begin{aligned}
h_n(s) 
&:=1 \notag\\
h_j(s)&:=s^{n-j}+\sigma_1s^{n-j-1}+\cdots+\sigma_{n-j},
\qquad j=1,\ldots,n-1, 
\end{aligned}
$$
As a consequence of the companion structure 
of $A$, we have
\begin{equation}\label{eq:Horner-canonical}
h_j(A)\eps_n=\eps_j.
\end{equation}
For $j=n$ this is obvious; backward induction follows from
\[
h_{j-1}(s)=s h_j(s)+\sigma_{n-j+1},\qquad
A\eps_j+\sigma_{n-j+1}\eps_n=\eps_{j-1}
\quad(2\le j\le n).
\]
Since $A$ is Hurwitz and
  $(A,\eps_n)$ is controllable the controllability Gramian
\begin{equation}\label{eq:G}
G=\int_0^\infty e^{At}\eps_n\eps_n^\top e^{A^\top t}\,dt
\end{equation}
is positive definite.
Notice that $G$ satisfies
\begin{equation}\label{eq:G-lyapunov}
AG+GA^\top=-\eps_n\eps_n^\top.
\end{equation}

Factor $G=LL^\top$ with $L$ real and nonsingular, and set
\[
S=L^{-1}AL,\qquad c=L^{-1}\eps_n.
\]
Multiplying \eqref{eq:G-lyapunov} by $L^{-1}$ and $L^{-\top}$ gives
\begin{equation}\label{eq:normalized-energy}
S+S^\top=-cc^\top.
\end{equation}
Thus $M=-S$ is accretive ($M+M^\top\succeq 0$).

Since $S$ is obtained from $A$ by a change of basis, it has the same characteristic polynomial $a(s)$ of $A$.
The adjugate--Horner
identity \cite[page 51]{Marcus-Minc} gives
\begin{equation}\label{eq:adj-Horner}
\adj(zI+M)=\adj(zI-S)=\sum_{j=1}^n z^{j-1}h_j(S).
\end{equation}
This identity follows by multiplication by $zI-S$ and the
Cayley--Hamilton theorem. 
Applying Lemma~\ref{lem:pairwise-adjugate}
with $F_{j-1}=h_j(S)$ gives
\begin{equation}\label{eq:Horner-pairwise}
h_k(S)h_l(S)^\top+h_l(S)h_k(S)^\top\succeq0.
\end{equation}

Finally, \eqref{eq:Horner-canonical} and commutation of the matrix
polynomials $h_j(A)$ with $e^{At}$ give
\begin{align}
P_{kl}
&=h_k(A)G\,h_l(A)^\top+h_l(A)G\,h_k(A)^\top\notag\\
&=L\left[h_k(S)h_l(S)^\top+h_l(S)h_k(S)^\top\right]L^\top.
\label{eq:P-congruence}
\end{align}
The conclusion follows from \eqref{eq:Horner-pairwise}.
\end{proof}

As a corollary of the previous result, we have the following
Theorem which is our main result.

\begin{theorem}\label{main-result}
Let $A$  be a real Hurwitz  companion matrix given by 
\eqref{comp-matrix} and $Q,R\in\Sn$.
Then the solution $P$ of (\ref{eq:lyap-primal}) is positive semidefinite
for all $Q\geq 0$ and  the solution $X$ of (\ref{eq:lyap-adjoint}) is entrywise nonnegative for all $R\succeq0$.
\end{theorem}
\begin{proof}
By using Proposition \ref{prop:lyapunov-primal-sc}, Property (\ref{eq:property-primal}) follows immediately.
In fact, any elementwise nonnegative matrix $Q\in\Sn$ can be written as the sum of a linear combination with nonnegative coefficients of simple matrices of the form $Q_{kl}$
given by \eqref{elementary-Q} and a nonnegative diagonal matrix
$D$ which is clearly positive semidefinite so that the solution 
$P_D$ of (\ref{eq:lyap-primal}) associated with $Q=D$ is positive semidefinite in view of the standard Lyapunov theory.
By linearity, we then have (\ref{eq:property-primal}).
Theorem \ref{theo:duality} allows to conclude that (\ref{eq:property-adjoint}) also holds.
\end{proof}

\section{A larger class of matrices $A$}\label{sec:main-res-ext}

We have established that if $A$ is a Hurwitz  companion matrix
then properties (\ref{eq:property-primal}) and  (\ref{eq:property-adjoint})
hold.
It is natural to ask whether these properties hold for a larger class of matrices.
The answer is yes.
Indeed, there is  a simple invariance principle.  Suppose that $\tilde{A}$ satisfies
\eqref{eq:property-primal} and let $T$ be an entrywise nonnegative and nonsingular matrix and set
\[
A=T^{-1}\tilde{A} T.
\]
Notice that, the class of Hurwitz  companion matrices is not invariant under this
transformation so that, even if $\tilde{A}$ is a Hurwitz matrix in companion form,
$A$ need not be in companion form.

Then
\begin{equation}\label{eq:similarity}
  \TA(Q)
  =T^{-1}\TAt(TQT^\top)T^{-\top}.
\end{equation}
Since $Q\geq 0$ implies $TQT^\top\geq 0$, equation
\eqref{eq:similarity} shows that also $A$  is such that 
\eqref{eq:property-primal} is satisfied.
In conclusion, we have the following result.
\begin{theorem}\label{main-result-gen}
Let $A$  be similar to a real Hurwitz  companion matrix
via an entrywise non-negative change of basis,  and $Q,R\in\Sn$.
Then the solution $P$ of (\ref{eq:lyap-primal}) is positive semidefinite
for all $Q\geq 0$ and  the solution $X$ of (\ref{eq:lyap-adjoint}) is entrywise nonnegative for all $R\succeq0$.
\end{theorem}

Unfortunately, while it is easy to see by direct computation that, in the case of $n=2$, this class
characterizes all the state matrices $A$ for which (\ref{eq:property-primal}) and  (\ref{eq:property-adjoint}) hold, for larger values of $n$, this is not the case.

Indeed consider the matrix
\begin{equation}\label{eq:A33}
  A:=
  \begin{bmatrix}
    2&4&-1\\
    -1&0&3\\
    -4&-8&-6
  \end{bmatrix}.
\end{equation}
We first show that it is such that (\ref{eq:property-primal}) and  (\ref{eq:property-adjoint}) are satisfied. 
For $k,l=1,2,3$, and $k<l$, consider the three matrices
$Q_{kl}$ as defined in \eqref{elementary-Q}. 
Let $P_{kl}$ be the associated solutions of
\eqref{eq:lyap-primal} with $Q=Q_{kl}$.
A direct  calculation gives
\[
 {\tiny P_{12}
  =\frac18
  \begin{bmatrix}
    44&-24&-8\\
    -24&25&-8\\
    -8&-8&16
  \end{bmatrix},\quad  P_{13}
  =\frac1{64}
  \begin{bmatrix}
    72&-36&0\\
    -36&33&-12\\
    0&-12&16
  \end{bmatrix},
\quad 
  P_{23}
  =\frac1{16}
  \begin{bmatrix}
    44&-24&-8\\
    -24&27&-8\\
    -8&-8&16
  \end{bmatrix}}
\]
which are readily seen to be positive definite.
Thus $A$ in \eqref{eq:A33} is such that properties (\ref{eq:property-primal}) and  (\ref{eq:property-adjoint}) are satisfied.

Nevertheless, $A$ is a Hurwitz matrix but  cannot be represented as
\[
  A=T^{-1}C T,
\]
where $C$ is a Hurwitz companion matrix and $T$ is invertible and
entrywise nonnegative.
In fact, the companion matrix having the same characteristic polynomial as $A$ is
\[
  C=
  \begin{bmatrix}
    0&1&0\\
    0&0&1\\
    -32&-12&-4
  \end{bmatrix}.
\]
Suppose, by contradiction, that
  $A=T^{-1}CT$ for some invertible entrywise nonnegative matrix $T$.  Equivalently,
there exists a nonsingular and nonnegative matrix $T$ such that $CT-TA=0.$
We can rewrite  this Sylvester equation
by the using the Kronecker product formula:
$$
[I\otimes C-A^\top\otimes I]{\rm vec}(T)=0,
$$
where, as already mentioned, ${\rm vec}(\cdot)$ is the linear operator mapping a matrix in a
vector obtained by stacking the columns of the matrix one on top of the other.
We can explicitly compute $[I\otimes C-A^\top\otimes I]$
and its kernel. This provides the following parametrization of the matrices $T$
for which $CT-TA=0$:
$$
{\rm vec}(T)=
\frac1{31}
\begin{bmatrix}
54&17&3\\
-96&18&5\\
-160&-156&-2\\
80&16&1\\
-32&68&12\\
-384&-176&20\\
31&0&0\\
0&31&0\\
0&0&31
\end{bmatrix}w, \qquad w\in\R^3.
$$
By considering the last $3$ rows, it is clear that 
to ensure that all the elements of $T$ are nonnegative
we must have $w\geq 0$.
On the other hand, from the third row it is clear that the only $w\geq 0$
for which all the elements of $T$ are nonnegative is $w=0$.
Hence the only $T\geq 0$ solving $CT-TA=0$ is $T=0$ which is clearly singular.

In conclusion,
\[
  A\neq T^{-1}CT
\]
for every Hurwitz companion matrix $C$ and every invertible entrywise
nonnegative matrix $T$.
Hence,  the  two-dimensional nonnegative-similarity
characterization does not extend to higher dimensions.

\subsection{Open neighbourhood of matrices $A$}
Consider the $2\times2$ real matrix
\begin{equation}\label{eq:concreteA}
  A:=
  \begin{bmatrix}
    \tfrac14&1\\
    -1&-\tfrac94
  \end{bmatrix}.
\end{equation}
To show that (\ref{eq:property-primal}) and  (\ref{eq:property-adjoint})
hold for this $A$ it is sufficient to 
observe that $\TA(Q_{12})\succeq 0,$ with
$Q_{12}=\begin{bmatrix}0&1\\1&0\end{bmatrix}$.
Indeed, we have 
\[
  \TA(Q_{12})
  =\tfrac{1}{14}
  \begin{bmatrix}
    36 &-9\\[1mm]
    -9&4
  \end{bmatrix}\succ 0.
\]
Moreover, since both pairs \((A,\varepsilon_1)\) and \((A,\varepsilon_2)\) are controllable,  we have $\TA(\varepsilon_1 \varepsilon_1^\top)\succ 0$
and $\TA(\varepsilon_2 \varepsilon_2^\top)\succ 0$ so that
the matrix $A$ in \eqref{eq:concreteA} satisfies the slightly stronger property
\begin{equation}\label{eq:strongp-prop}
  Q=Q^\top \geq0,\quad Q\ne0
  \quad\Longrightarrow\quad
  \TA(Q)\succ0.
\end{equation}
This stronger property is extremely interesting because
when it holds, it remains valid for each state matrix $A_\varepsilon$ in an open
neighborhood of the original matrix $A$.
To see this, consider the {\em finite} set of matrices 
$\mathcal G:=\{Q_{ii}, i=1,\dots, n\}\cup\{Q_{ij}; i,j=1,\dots, n, i<j\}$.
Let $Q\in\mathcal G$. If  $\mathcal T_A(Q)\succ 0$ then
there exists an open neighborhood $\mathcal{A}(Q)$ of $A$  such that 
for all $\hat{A}\in\mathcal{A}(Q)$,   $T_{\hat{A}}(Q)\succ 0$ still holds because the solution of a Lyapunov equation depends continuously on the Hurwitz
state matrix $A$. 
Since $\mathcal G$ is a finite set,  the intersection
$$
\mathcal{A}:=
\cap_{Q\in\mathcal G}
\mathcal{A}(Q)$$ 
is still an open neighbourhood of $A$  and if $A_\varepsilon\in \mathcal{A}$
we have
$$
Q\in\mathcal G\Longrightarrow \mathcal T_A(Q)\succ 0$$
which readily implies property \eqref{eq:strongp-prop}.
Consequently, all sufficiently
small perturbations of $A$ preserve Hurwitz stability and property \eqref{eq:strongp-prop}
and hence also properties (\ref{eq:property-primal}) and  (\ref{eq:property-adjoint}). 

\section{Control and filtering interpretations}
\label{sec:interpreta}

\subsection{Covariance normalization and dissipative dynamics}

The Gramian $G$ in \eqref{eq:G} is the stationary covariance of the
linear stochastic system
\[
dX_t=AX_t\,dt+\eps_n\,dW_t,
\]
where $W_t$ is a scalar standard Wiener process. In the coordinates
$\xi_t=L^{-1}X_t$, the stationary covariance is $I$, and the system is
\[
d\xi_t=S\xi_t\,dt+c\,dW_t.
\]
Its covariance balance is precisely
\eqref{eq:normalized-energy}. The associated unforced dynamics
$\dot\xi=S\xi$ satisfy
\[
\frac{d}{dt}\|\xi(t)\|^2
=\xi(t)^\top(S+S^\top)\xi(t)
=-|c^\top\xi(t)|^2.
\]
Thus Gramian normalization produces the dissipative realization needed
for Lemma~\ref{lem:pairwise-adjugate}.

\subsection{The auxiliary determinant as a damped system}

On setting $s=t=z$ in \eqref{eq:D-tau}, one obtains
\[
D_\tau(z,z)
=\det\!\left(z^2I+z(M+M^\top)+M^\top M+\tau xx^\top\right).
\]
This is the characteristic determinant of the damped second-order system
\[
\ddot u+(M+M^\top)\dot u+(M^\top M+\tau xx^\top)u=0.
\]
Its nonnegative energy
\[
\mathcal E(t)=\frac12\|\dot u(t)\|^2
+\frac12u(t)^\top(M^\top M+\tau xx^\top)u(t)
\]
satisfies
\[
\frac{d}{dt}\mathcal E(t)
=-\dot u(t)^\top(M+M^\top)\dot u(t)\le0.
\]
The stability argument in two independent variables strengthens this
energy reasoning enough to separate every pair of adjugate coefficients.

\subsection{The Gramian for the sum of two input directions}

For the proposed positive semidefinite forcing
\[
Q_{kkll}=(\eps_k+\eps_l)(\eps_k+\eps_l)^\top,
\]
define
\[
P_{kkll}=\int_0^\infty e^{At}Q_{kkll}e^{A^\top t}\,dt.
\]
Because $Q_{kk}=2\eps_k\eps_k^\top$, linearity gives
\[
P_{kkll}=\frac12(P_{kk}+P_{ll})+P_{kl}.
\]
Theorem~\ref{main-result} establishes the stronger comparison
\begin{equation}\label{eq:sum-Gramian}
 P_{kkll}\succeq\frac12(P_{kk}+P_{ll}). 
\end{equation}
Equivalently, for every scalar output direction $y\in\R^n$, the two
input-channel impulse responses have nonnegative $L^2$ inner product:
\[
\int_0^\infty
\bigl(y^\top e^{At}\eps_k\bigr)
\bigl(y^\top e^{At}\eps_l\bigr)\,dt
=\frac12y^\top P_{kl}y\ge0.
\]

\section{Conclusions}\label{sec:conclusions}
We have established a duality between entrywise nonnegativity and positive semidefiniteness for Lyapunov solution operators and proved that both positivity properties hold for every real Hurwitz companion matrix. This resolves the companion-matrix conjecture without restrictions on the reality or multiplicity of the eigenvalues. The proof combines controllability-Gramian normalization with a pairwise positivity theorem for adjugate coefficients of real accretive matrices, obtained through bivariate polynomial stability. Covariance and energy interpretations further connect these results with control and filtering. 
We have considerably extended the class of state matrices $A$ for which 
the Lyapunov solution operators satisfy these equivalent positivity properties
and we have seen that this class is quite vast.
However, we have shown that even this extension does not provide the full class
ensuring the positivity properties.
A natural direction for future research is to characterize this full class of  state matrices.

\end{document}